\def\confversion{0}
\def\ifconf{\ifnum\confversion=1}
\def\ifnotconf{\ifnum\confversion=0}

\documentclass[11 pt]{article}

\def\showauthornotes{1}
\def\showkeys{0}
\def\showdraftbox{0}

\usepackage{xspace,xcolor,enumerate}
\usepackage{amsmath,amssymb}
\usepackage{amsthm}
\usepackage{xfrac}
\usepackage[toc,page]{appendix}
\usepackage{thmtools}
\usepackage{thm-restate}
\usepackage{color,graphicx}
\usepackage{paralist}
\usepackage{thm-restate}

\ifnum\showkeys=1
\usepackage[color]{showkeys}
\fi

\definecolor{darkred}{rgb}{0.5,0,0}
\definecolor{darkgreen}{rgb}{0,0.35,0}
\definecolor{darkblue}{rgb}{0,0,0.55}

\usepackage[pdfstartview=FitH,pdfpagemode=UseNone,colorlinks,linkcolor=darkblue,filecolor=darkred,citecolor=darkgreen,urlcolor=darkred,pagebackref]{hyperref}

\usepackage[capitalise,nameinlink]{cleveref}
\usepackage[T1]{fontenc}
\usepackage{mathtools,dsfont,bbm}

\ifnum\showauthornotes=1
\newcommand{\Authornote}[2]{{\sf\small\color{red}{[#1: #2]}}}
\newcommand{\Authorcomment}[2]{{\sf \small\color{gray}{[#1: #2]}}}
\newcommand{\Authorfnote}[2]{\footnote{\color{red}{#1: #2}}}
\else
\newcommand{\Authornote}[2]{}
\newcommand{\Authorcomment}[2]{}
\newcommand{\Authorfnote}[2]{}
\fi

\ifnum\showdraftbox=1
\newcommand{\draftbox}{\begin{center}
  \fbox{
    \begin{minipage}{2in}
      \begin{center}
        \begin{Large}
          \textsc{Working Draft}
        \end{Large}\\
        Please do not distribute
      \end{center}
    \end{minipage}
  }
\end{center}
\vspace{0.2cm}}
\else
\newcommand{\draftbox}{}
\fi

\newtheorem{theorem}{Theorem}[section]

\newtheorem{definition}[theorem]{Definition}
\newtheorem{lemma}[theorem]{Lemma}
\newtheorem{remark}[theorem]{Remark}
\newtheorem{proposition}[theorem]{Proposition}
\newtheorem{corollary}[theorem]{Corollary}

\newtheorem{fact}[theorem]{Fact}

\crefname{nothing}{}{}
\crefformat{nothing}{(#2#1#3)}

\def\FullBox{\hbox{\vrule width 6pt height 6pt depth 0pt}}

\def\qed{\ifmmode\qquad\FullBox\else{\unskip\nobreak\hfil
\penalty50\hskip1em\null\nobreak\hfil\FullBox
\parfillskip=0pt\finalhyphendemerits=0\endgraf}\fi}

\def\qedsketch{\ifmmode\Box\else{\unskip\nobreak\hfil
\penalty50\hskip1em\null\nobreak\hfil$\Box$
\parfillskip=0pt\finalhyphendemerits=0\endgraf}\fi}

\def\to{\rightarrow}
\def\eps{\varepsilon}
\def\epsilon{\varepsilon}

\def\eps{\epsilon}

\def\phi{\varphi}
\def\cal{\mathcal}

\newcommand{\mper}{\,.}
\newcommand{\mcom}{\,,}

\newcommand{\R}{{\mathbb R}}

\newcommand{\N}{{\mathbb{N}}}
\newcommand{\Z}{{\mathbb Z}}

\newcommand{\F}{{\mathbb F}}

\usepackage{nicefrac}

\makeatletter
\def\norm#1{
  \@ifnextchar\bgroup
   {\normalpnorm{#1}}
   {\defaultnorm{#1}}
}
\def\defaultnorm#1{
    \renderNorm{#1}{}
}
\def\normalpnorm#1#2{
   \@ifnextchar\bgroup
   {\banaxnorm{#1}{#2}}
   {\renderNorm{#2}{#1}}
}

\def\banaxnorm#1#2#3{
    \renderNorm{#3}{#1\to#2 }

}

\def\renderNorm#1#2{
    \@ifnextchar^
    {\fixExponent{#1}{#2}}
    {\ensuremath{\mathchoice
        {\lVert #1 \rVert_{#2}}
        {\lVert #1 \rVert_{#2}}
        {\lVert #1 \rVert_{#2}}
        {\lVert #1 \rVert_{#2}}}
    }
}
\def\fixExponent#1#2^#3{
    \ensuremath{{\lVert #1 \rVert^{#3}_{#2}}}
}

\def\enorm#1#2{
   \@ifnextchar\bgroup
   {\norm{L_{#1}}{L_{#2}}}
   {\norm{L_{#1}}{#2}}
}

\def\cnorm#1#2{
   \@ifnextchar\bgroup
   {\norm{\ell_{#1}}{\ell_{#2}}}
   {\norm{\ell_{#1}}{#2}}
}

\makeatother

\newcommand{\mysmalldot}[2]{\ensuremath{\langle #1, #2 \rangle}}

\newcommand{\Esymb}{\mathbb{E}}
\newcommand{\Psymb}{\mathbb{P}}

\makeatletter
\def\Pr#1{
    \ProbabilityRender{\Psymb}{#1}
}

\def\Ex#1{
    \ProbabilityRender{\Esymb}{#1}
}

\def\ProbabilityRender#1#2{
  \@ifnextchar\bgroup
  {\renderwithdist{#1}{#2}}
   {\singlervrender{#1}{#2}}
}
\def\singlervrender#1#2{
   \ensuremath{\mathchoice
       {{#1}\left[ #2 \right]}
       {{#1}[ #2 ]}
       {{#1}[ #2 ]}
       {{#1}[ #2 ]}
   }
}
\def\renderwithdist#1#2#3{
   \@ifnextchar\bgroup
   {\superfancyrender{#1}{#2}{#3}}
   {\ensuremath{\mathchoice
      {\underset{#2}{#1}\left[ #3 \right]}
      {{#1}_{#2}[ #3 ]}
      {{#1}_{#2}[ #3 ]}
      {{#1}_{#2}[ #3 ]}
     }
   }
}
\def\superfancyrender#1#2#3#4#5{
   \ensuremath{\mathchoice
      {\underset{#1}{{#1}}\left#4 #3 \right#5}
      {{#1}_{#2}#4 #3 #5}
      {{#1}_{#2}#4 #3 #5}
      {{#1}_{#2}#4 #3 #5}
   }
}
\makeatother

\newfont{\inhead}{eufm10 scaled\magstep1}

\newcommand{\calF}{{\cal F}}

\newcommand{\calO}{{\cal O}}

\DeclareMathOperator\supp{Supp}

\DeclareMathOperator*{\argmin}{\arg\!\min}

\newcommand{\ceil}[1]{\ensuremath{\left\lceil #1 \right\rceil}}
\newcommand{\floor}[1]{\ensuremath{\left\lfloor #1 \right\rfloor}}

\newcommand{\classfont}[1]{\textsf{#1}}

\newcommand{\BPP}{\classfont{BPP}\xspace}
\newcommand{\classP}{\classfont{P}\xspace}
\newcommand{\NP}{\classfont{NP}\xspace}

\usepackage[colorinlistoftodos,textsize=tiny,textwidth=2cm,color=red!25!white,obeyFinal]{todonotes}

\newcommand{\SSS}{\mathbb{S}}

\newcommand{\bT}{\overline{T}}

\newcommand{\dsp}[2]{\mathbb{S}^{#2}_{#1}}

\newcommand{\LCD}{\mathrm{LCD}}

\newcommand{\dist}{\mathrm{dist}}
\newcommand{\spn}{\mathrm{span}}

\newcommand{\tilu}{\widetilde{u}}

\newcommand{\al}{\alpha}
\newcommand{\cN}{\mathcal{N}}

\newcommand{\incomp}{\mathrm{Incomp}}

\newcommand{\comp}{\mathrm{Comp}}

\crefformat{equation}{(#2#1#3)}

\usepackage[top=1in, bottom=1in, left=1in, right=1in]{geometry}

\allowdisplaybreaks
\begin{document}

\title{\textbf{Inapproximability of Finding Sparse Vectors \\ in Codes, Subspaces, and Lattices}\footnote{This document is a merger of \cite{BL24} which proved the result for the reals and a follow-up work~\cite{BGR25} which 
adapts the reduction of \cite{BL24} to the case of finite fields while simplifying and derandomizing the reduction.} }

\author{
Vijay Bhattiprolu 
\thanks{University of Waterloo. Email: {\tt vbhattip@uwaterloo.ca}. Supported by NSERC Discovery Grant 50503-11559 500.}
\and
Venkatesan Guruswami
\thanks{Simons Institute for the Theory of Computing, and UC Berkeley. Email: {\tt venkatg@berkeley.edu}. Research supported in part by NSF grants CCF-2228287 and CCF-2211972 and a Simons Investigator award. }
\and 
Euiwoong Lee
\thanks{University of Michigan. Email: {\tt euiwoong@umich.edu}. Supported in part by NSF grant CCF-2236669 and Google.}
\and
Xuandi Ren\thanks{UC Berkeley. Email: \texttt{xuandi\_ren@berkeley.edu}. Supported in part by NSF grant CCF-2228287.} 
}

\setcounter{page}{0}
\date{}

\maketitle
\draftbox
\thispagestyle{empty}

\begin{abstract}
    Finding sparse vectors is a fundamental problem that arises in several contexts including codes, subspaces, and lattices. In this work, we prove strong inapproximability results for all these variants using a novel approach that even bypasses the PCP theorem.
    Our main result is that it is NP-hard (under randomized reductions) to approximate the sparsest vector in a real subspace within any constant factor;
    the gap can be further amplified using tensoring. Our reduction has the property that there is a Boolean solution in the completeness case.
    As a corollary, this immediately recovers the state-of-the-art inapproximability factors for the shortest vector problem (SVP) on lattices. Our proof extends the range of $\ell_p$ (quasi) norms for which hardness was previously known, from $p\geq 1$ to all $p\geq 0$, answering a question raised by (Khot, JACM 2005).

\smallskip
    Previous hardness results for SVP, and the related minimum distance problem (MDP) for error-correcting codes,
    all use lattice/coding gadgets that have an abundance of codewords in a ball of radius smaller than the minimum distance. In contrast, our reduction only needs many codewords in a ball of radius slightly larger than the minimum distance. This enables an easy
    derandomization of our reduction for finite fields, giving a new elementary proof of deterministic hardness for MDP.
    We believe this weaker density requirement might offer a promising approach to showing deterministic hardness of SVP, a long elusive goal.
The key technical ingredient underlying our result for real subspaces is a proof that in the kernel of a random Rademacher matrix, the support of any two linearly independent vectors have very little overlap.

    \smallskip
  A broader motivation behind this work is the development of inapproximability techniques for problems over the reals.
    Analytic variants of sparsest vector have connections to small set expansion, quantum separability and polynomial maximization over convex sets, all of which appear to be out of reach of current PCP techniques. We hope that the approach we develop could enable progress on some of these problems.
\end{abstract}

\newpage

\pagenumbering{arabic}
\setcounter{page}{1}

\section{Introduction}\label{sec:intro}
Let $\F$ be a field and let $\norm{0}{x}$ denote the Hamming weight (number of nonzero entries) of a vector $x\in \F^n$.
For $\F=\R$ and $p>0$, let $\norm{p}{x} := (\sum_i |x_i|^p)^{1/p}$ denote the $\ell_p$-(quasi) norm.
For a subset $U$ of $\F^n$, finding the sparsest (or shortest) nonzero vector, i.e.,
\begin{equation}
    \label{spvec}
    \argmin_{x\in U\setminus \{0\}}\norm{}{x},
\end{equation}
is a fundamental problem that arises in several contexts.
If $U$ is a subspace over a finite field $\F_q$, and $\norm{}{\cdot}=\norm{0}{\cdot}$,
one obtains the minimum distance problem (which we denote as MDP($\F_q$)).
If $U$ is an integer lattice, and $\norm{}{\cdot}=\norm{p}{\cdot}$, \cref{spvec} captures the shortest vector problem (SVP$_p$). Each of these tasks is foundational in the study of error correcting codes and lattice based cryptography, respectively.
If $U$ is a real subspace, and $\norm{}{\cdot}=\norm{0}{\cdot}$, \cref{spvec} captures the nullspace problem -- a homogeneous\footnote{where we say an optimization problem is homogeneous if the feasible region is closed under scaling} variant of the sparse recovery problem, which we denote as MDP($\R$) and has connections to robust subspace recovery~\cite{candes2011robust, lerman2018overview},
dictionary learning~\cite{elad2010sparse, BKS14, BKS15}, sparse blind deconvolution~\cite{zhang2017global, kuo2019geometry} and many other areas; we refer the reader to the survey~\cite{qu2020finding} from the nonconvex optimization literature.

In this work, we prove strong hardness of approximation results for all of the above sparse/short vector problems using a novel approach that even bypasses the PCP theorem.
Our main contribution is a new inapproximability result over the reals:
\begin{theorem}
\label{thm:nsp}
No polynomial time algorithm can given a linear subspace $V \subseteq \R^n$ and $s \in \N$,  distinguishes between the following cases
\begin{compactenum}[]
    \item (YES) there exists nonzero $x \in V \cap \{ 0, 1 \}^n$ with $\norm{0}{x} \leq s$;
    \item (NO) every $x \in V \setminus \{ 0 \}$ satisfies $\norm{0}{x} \geq \gamma \cdot s$,
\end{compactenum}
\medskip

\begin{compactenum}[(a)]
    \item assuming $\NP\not\subseteq \BPP$ when $\gamma>1$ is any constant;
    \item assuming $\NP\not\subseteq \classfont{BPTIME}(2^{\log^{O(1)} n})$ when $\gamma =2^{\log^{1-\eps}n}$ for any fixed $\eps>0$;
    \item assuming $\NP\not\subseteq \bigcap_{\delta>0} \classfont{BPTIME}(2^{n^{\delta}})$ when $\gamma=n^{c/\log \log n}$ for some fixed $c>0$.
\end{compactenum}
\end{theorem}
\noindent
Previously,  only NP-hardness of exact optimization was known~\cite{mccormick1983combinatorial, coleman1986null}. The best known approximation algorithm achieves an $O(n / \log n)$-approximation~\cite{berman2001approximating}.

The reduction in \cref{thm:nsp} has the additional structural property that the solution is Boolean in the completeness case.
As a result of these strong completeness and soundness guarantees, we obtain inapproximability for SVP$_p$ as an immediate corollary:
\begin{theorem}
\label{thm:svp}
Fix $p\in [0,\infty)$. No polynomial-time algorithm can given a lattice $L \subseteq \Z^n$ and $s \in \N$, distinguishes between the following cases
\begin{compactenum}[]
        \item (YES) there exists a nonzero vector $x \in L \cap \{ 0, 1 \}^n$ with $\| x \|_p = s^{1/p}$;

        \item (NO) every $x \in L \setminus \{ 0 \}$ satisfies $\| x \|_p \geq \gamma^{1/p} \cdot s^{1/p}$,
\end{compactenum}
\medskip

\begin{compactenum}[(a)]
    \item assuming $\NP\not\subseteq \BPP$ when $\gamma>1$ is any constant;
    \item assuming $\NP\not\subseteq \classfont{BPTIME}(2^{\log^{O(1)} n})$ when $\gamma =2^{\log^{1-\eps}n}$ for any fixed $\eps>0$;
    \item assuming $\NP\not\subseteq \bigcap_{\delta>0} \classfont{BPTIME}(2^{n^{\delta}})$ when $\gamma=n^{c/\log \log n}$ for some fixed $c>0$.
\end{compactenum}
\end{theorem}

Above gives a new simple proof of state of the art inapproximability factors for SVP$_p$~\cite{khot2005hardness,haviv2007tensor,micciancio2012inapproximability}\footnote{albeit with two-sided error, whereas \cite{haviv2007tensor,micciancio2012inapproximability} only have one-sided error}, and
also expands the range of $p$ for which hardness was previously known from $p\geq 1$ to $p\geq 0$.
\cite{khot2005hardness}\footnote{see section ``wishful thinking''} asked whether one can obtain hardness of SVP$_0$ where there is a Boolean solution in the completeness case. We answer this question in the affirmative.

\medskip\noindent\textbf{Tensoring.} Gap amplification for SVP in \cite{khot2005hardness,haviv2007tensor} is highly nontrivial. The tensoring was made much cleaner in \cite{micciancio2012inapproximability}, by making use of a measure of length that interpolates between $\ell_p$ norms and Hamming weight. A feature of our reduction is that tensoring is trivial, since we work directly with the Hamming weight while tensoring, and deduce SVP hardness at the end as a corollary.

\medskip\noindent\textbf{Deterministic Hardness of MDP($\F_q$).}
The reduction in \cref{thm:nsp} takes on a particularly simple form in the special case of finite fields and we are able to derandomize it quite easily to obtain a new elementary proof of deterministic hardness for MDP($\F_q$):
\begin{theorem}
\label{thm:mdp}
Fix any finite field $\F_q$. No polynomial-time algorithm can given a linear subspace $V \subseteq \F_q^n$ and $s \in \N$, distinguishes between the following cases
\begin{compactenum}[]
    \item (YES) there exists nonzero $x \in V$ with $\norm{0}{x} \leq s$;
    \item (NO) every $x \in V \setminus \{ 0 \}$ satisfies $\norm{0}{x} \geq \gamma \cdot s$,
\end{compactenum}

\medskip

\begin{compactenum}[(a)]
    \item assuming $\NP\neq \classP$ when $\gamma>1$ is any constant;
    \item assuming $\NP\not\subseteq \classfont{DTIME}(2^{\log^{O(1)} n})$ when $\gamma =2^{\log^{1-\eps}n}$ for any fixed $\eps>0$;
    \item assuming $\NP\not\subseteq \bigcap_{\delta>0} \classfont{DTIME}(2^{n^{\delta}})$ when $\gamma=n^{c/\log \log n}$ for some fixed $c>0$.
\end{compactenum}

\end{theorem}

Deterministic inapproximability of MDP($\F_q$) was open for a long time until Cheng and Wan~\cite{cheng2009deterministic} derandomized a reduction of Dumer, Micciancio, and Sudan~\cite{dumer2003hardness}
by giving a deterministic construction of the locally dense gadget over any finite field.
In particular they give an explicit Hamming ball of radius $0.67d$ that contains exponentially many codewords of an explicit code of distance $d$.
The proof is fairly deep, making use of Weil's character sum estimate.
Austrin and Khot~\cite{AK14} gave a much simpler proof of deterministic hardness of MDP($\F_q$) by making use of tensor codes.
Building on this, Micciancio~\cite{M14} proved that for the special case of $\mathbb{F}_2$, the tensoring of any base code with large enough distance yields a locally dense gadget.

\medskip\noindent\textbf{Overview of Reduction.}
All our results are ``PCP-free" and reduce from the NP-hard problem of solving a system of quadratic equations over the concerned field $\F$ (a finite field or the reals).\footnote{Sometimes we consider homogeneous systems, and technically we work with a promise variant where YES instances have a Boolean solution whereas NO instances lack any solution over the field $\F$.} At a high level, we consider the subspace $X$ of symmetric matrices which are solutions to this quadratic system viewed as a linear system. Rank $1$ solutions $X = x x^T$ correspond to solutions $x$ to the original quadratic system. For finite fields, if we encode $x$ via a suitable gadget code where all nonzero codewords have roughly the same Hamming weight $d$, then we would have a low weight solution (of weight $\approx d$) for YES instances.

For the soundness, we need to handle spurious higher rank solutions. But even rank $2$ matrices have much higher Hamming weight, because in any linear code of minimum distance $d$, the minimum support size of a $2$-dimensional subspace, which is called the
2nd Generalized Hamming Weight \cite{Wei91,TV95}, is at least $\alpha \cdot d$ for $\alpha$ bounded away from $1$, specifically $\alpha = 1+1/|\F|$.
This pretty much gives the MDP hardness (one gets a gap $\approx 1+ 1/|\F|$ bounded away from $1$ that can be amplified by tensoring the code).

 The reduction above needed two features from the gadget code: (Weak Local Density): lots of codewords of Hamming weight $\approx d$ (in fact above we stipulated all codewords had this property and this stronger guarantee is achievable, but a slightly more complicated reduction works with merely an abundance of such codewords), and (Non-Overlap):  the union of support of two linearly independent codewords has size $\ge \alpha \cdot d$ for $\alpha$ bounded away from $1$.
 For finite fields, (Non-Overlap) is automatically true for any code with $\alpha = 1+1/|\F|$, and a simple code construction, namely a low rate Reed-Solomon code concatenated with the Hadamard code, achieves (Weak Local Density).

 Over reals, these two properties are non-trivial to achieve simultaneously. We prove, with a delicate chaining argument, that the kernel of a  matrix of suitable dimensions with i.i.d. $\pm 1$ random entries satisfies (Non-Overlap) with high probability (in fact with $\alpha \approx 2$). We achieve this by establishing a close connection between order-$2$ Hamming weight of the kernel, and the probability that a randomly
 signed sum of two-dimensional vectors $\sum_i \xi_i \cdot v_i$ (where $\{\xi_i\}$ are i.i.d. $\pm 1$) lies in a small ball around the origin.
 This latter question is central in Littlewood-Offord theory, and we are able to estimate this probability by utilizing a powerful
 result of Rudelson and Vershynin~\cite{rudelson2009smallest}
 who estimate the small ball probability in terms of a certain measure of the arithmetic
 structure of $\{v_1,v_2,\dots\}$.

 As for (Weak Local Density), it has already been shown in \cite{FSSZ23} (along with minimum distance estimates that come in handy for us).
 Plugging this ``code over reals'' into the above framework then gives us the hardness claimed in \cref{thm:nsp} via a \emph{randomized} reduction.

\medskip\noindent\textbf{A Brief History of the Homogenization Framework.}
When $U$ is an affine subspace or an affine lattice, we obtain non-homogeneous variants of MDP and SVP that are known as the
Nearest Codeword Problem (NCP) and the Closest Vector Problem (CVP) respectively.
Inapproximability of NCP and CVP can be deduced quite easily from the PCP theorem. In fact, this is one of the early applications
of the PCP theorem~\cite{ABSS93}. In contrast, the homogeneous variants MDP and SVP resisted efforts for decades. The usual reduction paradigm in PCP theory of replacing variables and constraints by constant sized gadgets does not appear to work.

A long line ~\cite{Adleman95,Var97,micciancio2001shortest, dumer2003hardness, khot2005hardness, cheng2009deterministic, micciancio2012inapproximability, AK14,M14} of important works culminated in strong hardness results for SVP and MDP.
All of these works follow the template of reducing from a non-homogeneous problem (CVP or NCP) whose hardness is established via the PCP theorem,
and then reducing to its homogeneous version (say SVP or MDP) by embedding solutions to the non-homogeneous problem inside a locally dense lattice/code.

\medskip\noindent\textbf{Eschewing Homogenization and Weakening the Local Density Requirement.}
We break from this paradigm in a few ways. Over finite fields, in contrast to \cite{AK14} that use tensor codes to reduce from NCP to MDP, we use tensor codes to directly linearize a system of quadratic equations. We also use a different encoding scheme for the solution $x$ of the starting hard problem. We encode $x\in\F_q^n$ as $xx^T$, whereas \cite{AK14} uses an indicator matrix $Z\in \{0,1\}^{nq\times nq}$ of $xx^T$.
This choice greatly simplifies the requirements of the coding gadget for $q\geq 3$. In \cite{AK14} the gadget (for fields of size $\geq 3$) uses Viola’s~\cite{Viola09} construction of a pseudorandom generator for low degree polynomials.

Over the reals we encode $x\in\{0,1\}^n$ as $yy^T$ for some $y\in C\cap \{0,1\}^N$ satisfying $Ty=x$ where $C$ and $T$ are an appropriately chosen code and linear projection respectively.
By Sauer-Shelah lemma, the existencee of such a map $T$ is equivalent to a a weak version of local density for $C$, i.e., $C$ contains exponentially many codewords in a larger radius than the distance.
We also deduce hardness for the non-homogeneous variants (\Cref{thm:ncp}) as a quick corollary of the homogeneous hardness, thereby reversing the usual chain of reductions in the area.

\medskip\noindent\textbf{Towards Derandomized Hardness of SVP.}
It remains open to derandomize our reduction, and in particular obtain the long elusive NP-hardness of SVP. Our approach might be more amenable to derandomization as it requires lots of vectors of sparsity $(1+\epsilon) d$ instead of many vectors of sparsity $\le (1-\epsilon) d$ around some nonzero center. The fact that the radius can exceed $d$ can make such objects easier to construct deterministically, as was indeed the case for codes. We however also need the code to have gap between the order one and order two Hamming weights.
Using the same reduction, the local density requirement can be further weakened at a cost of demanding more from higher order Hamming weight and the hardness of quadratic equations. This seems a compelling direction for future investigation.

\medskip\noindent\textbf{Analytic Sparsity Problems.}
Another important motivation behind this work is the development of techniques to prove inapproximability results for problems over the reals.
Sparsest vector in a subspace is one of a long list of problems, including polynomial maximization over convex sets, quantum separability,
maximizing  $\|\cdot\|_q/\|\cdot\|_p$ (an analytic notion of sparsity when $q>p$) over a subspace,
small set expansion, densest $k$-subgraph, sparse PCA, low rank matrix completion, tensor PCA/rank, etc., that are resistant to the ``local gadget"
\footnote{not to be confused with locally dense gadget -- local gadgets are often of constant size whereas locally dense gadgets are often polynomial size}
reduction paradigm in PCP theory. Informally, this is because such reductions from a PCP fatally contain very
sparse solutions\footnote{Assuming the starting PCP doesn't come with appropriate expansion vs. smoothness properties. Such PCPs appear out of
reach of current techniques.}.
In \cite{BGGLT23} inapproximability for $p\rightarrow q$ operator norm when $2<p<q$, is shown using a global reduction -- it uses a classical embedding result from convex geometry. In this work we have shown yet another example that global geometric reductions are successful for problems over reals.

Related problems remain wide open and have important implications.
It was shown in \cite{BBHKSZ12} that the small set expansion of a graph $G$ can be cast as finding the
sparsest vector that is close (in $\ell_2$ norm) to the top eigenspace of $G$.
It is also shown~\cite{HM13,BBHKSZ12} that the hardness of approximately computing the $2\rightarrow 4$ sparsity of a subspace is closely related to $\mathsf{QMA}=\mathsf{NEXP}$, which is a longstanding open problem in quantum information.
In \cite{bhattiprolu2021framework}, it was shown that NP-hardness of $p\rightarrow 2$ sparsest vector (for all $p<2$)
would lead to an NP-based near-characterization of the convex sets over which quadratics can be approximately maximized.
We believe our work provides a promising new line of attack on hardness of $p\rightarrow q$ sparsest vector (for all $p<q$).

\section{Preliminaries and Proof Overview}
\label{sec:pre}
Let $\F$ be any field. Any homogeneous $n$-variate quadratic polynomial $p:\mathbb{F}^n\to\mathbb{F}$ may be written in the form $p(x) = \sum_{i,j\in [n]}Q[i,j]x_ix_j$ for some coefficient matrix $Q=[Q[i,j]]_{i,j\in [n]} \in \mathbb{F}^{n\times n}$. For $X\in \mathbb{F}^{n\times n}$, let $Q(X):= \sum_{i,j\in [n]}Q[i,j] X[i,j]$, so that $p(x) = Q(xx^T)$ for all $x\in\mathbb{F}^n$.
It will be convenient for us to encode homogeneous quadratics by their coefficient matrix Q.

\subsection{Quadratic Equations Hardness}

To prove hardness of MDP($\F_q$) and MDP($\R$), we will require hardness of homogeneous and non-homogeneous variants of
satisfiability of quadratic equations respectively:
\begin{proposition}(NP-Hardness of Quadratic Equations)~\\
\label{prop:quadeq}
Let $\F$ be any field. Given a system of quadratic equations over $\F^n$ of the form $\{Q_\ell(xx^T) = b_{\ell}\}_{\ell \in [m]}$
(resp. $\{Q_\ell(xx^T) = 0\}_{\ell \in [m]}$), it is NP-hard to distinguish between the following two cases:
\begin{itemize}
\itemsep=0ex
        \item (YES) There exists $x \in \{ 0, 1 \}^n\setminus \{0\}$ satisfying all $m$ equations.

        \item (NO) There does not exist $x \in \F^n$ (resp. $x \in \F^n\setminus \{0\}$)  satisfying all $m$ equations.
\end{itemize}
\end{proposition}
Note that the above completeness guarantees a solution $x \in \{0,1\}^n$ whereas the soundness rules out $x \in \F^n$ -- it is thus a promise problem.

\Cref{prop:quadeq} is proved via reduction from the circuit satisfiability (\textsc{Circuit-SAT}) problem. The proof follows the standard template for exact NP-Hardness results, and we defer it to \Cref{sec:quadeq}.

\subsection{Tensor Codes and Distance Amplification}

We show hardness of MDP by first generating a constant factor gap and then using the standard observation that the minimum distance of a code is multiplicative under the usual tensor product operation, which we prove below for completeness.

The tensor product of two subspaces $U\subseteq \mathbb{F}^n$ and $V\subseteq \mathbb{F}^m$, denoted by $U\otimes V$ may be defined as the space of matrices $M\in \mathbb{F}^{n \times m}$ such that every row of $M$ lies in $V$ and every column of $M$ lies in $U$.
Let $d(U)$ denote the minimum distance of a subspace $U$, i.e.,
$d(U) := \min_{u\in U\setminus \{0\}} \|u\|_0$.
Then we have
\begin{restatable}{fact}{tensor}
\label{mult:dist}
    For any subspaces $U\subseteq \mathbb{F}^n$ and $V\subseteq \mathbb{F}^m$,
    $d(U\otimes V) = d(U) \cdot d(V)$.
\end{restatable}
We defer the proof of this standard fact to ~\cref{appendix:tensoring}. Applying \cref{mult:dist} inductively yields:
\begin{fact}
\label{tensoring}
    For any subspace $U\subset \mathbb{F}^n$ and any $t\in \mathbb{N}$,
    $d(U^{\otimes t}) = d(U)^t $.
\end{fact}

\subsection{Hamming Weight of Rank $\geq 2$ Elements of a Tensor Code}
Recall that in \Cref{mult:dist}, the upper bound $d(C)^2$ on $d(C\otimes C)$ (taking $U=V=C$) is attained by a rank-1 matrix.
The following result implies that for any linear code over small fields, codewords of rank $\geq 2$ in $C\otimes C$ have Hamming weight significantly larger than the minimum distance. Austrin and Khot~\cite{AK14} were the first to realize its utility in the context of hardness of the minimum distance of codes. The $\mathbb{F}_2$ case of the below was stated and used earlier in \cite{GGR} toward list decoding tensor product codes.

\begin{lemma}[Rank-$2$ Elements of Tensor Codes have Large Hamming Weight~\cite{AK14}]
\label{lem:rk2hwfq}
    For every subspace $C\subseteq \mathbb{F}_q^n$ and every $M\in C\otimes C$ of rank at least $2$, we have $\|M\|_0 \geq \left(1+\frac{1}{q}\right)\cdot d(C)^2$.
\end{lemma}

The above lemma follows from the fact that the support of a 2-dimensional subspace of a linear code, which is called the 2nd Generalized Hamming weight in the literature \cite{Wei91, TV95}, is larger than the minimum distance by a constant factor, a feature which we capture by the following definition.

\def\supp{\text{supp}}

\begin{definition}[Non-Overlap]\label{def:overlap}
Let $\mathbb{F}$ be any field. A subspace $C \subset \mathbb{F}^n$ is said to be $\alpha$-non-overlapping if for some $\alpha \ge 1$ if for any $u,v \in C$ that are linearly independent over $\mathbb{F}$, we have
\[ |\supp(u) \cup \supp(v) | \ge \alpha \cdot d(C) \ . \]
We call $\alpha$ the non-overlapping coefficient of $C$.
\end{definition}

Equipped with the above definition, we now state and prove a generalization of \cref{lem:rk2hwfq} abstracted through the $\alpha$-non-overlapping property.

\begin{lemma}
\label{lem:rk2hw}
Let $\mathbb{F}$ be any field and $C\subseteq \mathbb{F}^n$ be an arbitrary subspace that is $\alpha$-non-overlapping for some $\alpha \ge 1$. Then every $M\in C\otimes C$ of rank at least $2$ satisfies $\|M\|_0 \geq \alpha\cdot d(C)^2$.
\end{lemma}

\begin{proof}
    Since $M$ is of rank at least $2$, there are two linearly independent
    columns, that  have joint support of size at least $\alpha \cdot d(C)$ by the $\alpha$-non-overlapping property of $C$. Thus at least $\alpha \cdot d(C)$ rows of $M$ are non-zero, and since they lie inside $C$, each of these rows has at least $d(C)$ non-zero entries.
    Thus $\|M\|_0\ge \alpha \cdot d(C)^2$.
\end{proof}

The following lemma shows one can take $\alpha = 1+\frac{1}{q}$ for any subspace over $\mathbb F_q$, and this is in general tight as evidenced by the Hadamard code.

\begin{restatable}{lemma}{overlap}

\label{lem:overlap}
    Let $C$ be an arbitrary subspace over $\F_q$. Then $C$ is $\left(1+\frac{1}{q}\right)$-non-overlapping.
\end{restatable}

The proof of \Cref{lem:overlap} is quite simple and can be found in \Cref{sec:overlap}, and \Cref{lem:rk2hwfq} follows from plugging \Cref{lem:overlap} into \Cref{lem:rk2hw}.

\paragraph{Non-Overlap for Real Codes.}
As mentioned above, the non-overlapping
coefficient of a code over $\mathbb{F}_q$ can be at most $1+\frac{1}{q}$, and thus approaches $1$ for large fields (and in fact even equal $1$ when the field is the reals).  However, this is only true for the ``worst" subspaces and one might expect that typical subspaces of suitable dimension can have much larger non-overlapping coefficients. As one of our main technical results, we show in \cref{random:properties} that, for any fixed $\eps>0$, the kernel of a random Rademacher matrix over the reals is $(2-\epsilon)$-non-overlapping.

\section{PCP-Free Deterministic Reduction for MDP($\F_q$)}
\label{sec:mdp}
\paragraph{Rank-1 Testing over $\F_q$ via $\varepsilon$-Balanced Tensor Codes.}
At the heart of our reduction that generates constant factor hardness for MDP($\F_q$) is the observation that the connection between rank and Hamming weight in \cref{lem:rk2hw} can be made two-sided, assuming all codewords of the base code being tensored have similar hamming weight.
Such a code is called $\varepsilon$-balanced.
We observe that any codeword in $C\otimes C$ for an $\varepsilon$-balanced code $C$ has near-minimum Hamming weight if and only if it is rank-1.
We formalize this discussion below:

\begin{definition}\label{def:eps-balanced}
    For any constant $\varepsilon>0$, we say a linear error-correcting code with encoding map $G:\mathbb F_q^n \to \mathbb F_q^N$ and minimum distance $d$ is $\varepsilon$-balanced\,\footnote{
    The usual definition of $\varepsilon$-balanced is for binary linear codes and has
    the additional requirement that the minimum distance is $N/2(1-\Theta(\varepsilon))$.
    For our purposes, the minimum distance is unconstrained. We abuse terminology
    and continue to use the term $\varepsilon$-balanced.
    We also use this terminology for larger fields.
    }
    if the Hamming weight of every nonzero codeword lies in the range $[d,(1+\varepsilon)d]$.
\end{definition}
We remark that $\varepsilon$-balanced codes satisfy a weak version of local density, namely a Hamming ball of radius $(1+\varepsilon)d$ contains exponentially many (in fact all) codewords.

\paragraph{Constructing $\eps$-Balanced Codes.}
$\varepsilon$-balanced codes can be easily constructed by concatenating a Reed-Solomon code with the Hadamard code \cite{AGHP92}. Specifically, we have the following lemma:
\begin{lemma} \label{lem:eps-biased}
    For any constant $\varepsilon>0$, any finite field $\mathbb F_q$, and any $n \in \mathbb N$, there exists $N \le (qn/\varepsilon)^2$ and a linear code $G:\mathbb F_q^n \to \mathbb F_q^{N}$ with minimum distance at least $d = (1-\varepsilon)(1-\frac{1}{q})N$, satisfying
    $$\|G(x)\|_0 \in \left[(1-\varepsilon)\left(1-\frac{1}{q}\right)N,\left(1-\frac{1}{q}\right)N\right], \forall x \in \mathbb F_q^n \setminus \{0\} .$$
\end{lemma}

Note that when $\varepsilon<\frac{1}{2}$, we have $\left[(1-\varepsilon)\left(1-\frac{1}{q}\right)N,\left(1-\frac{1}{q}\right)N\right] \subseteq [d,(1+2\varepsilon) d]$, which means the code is $(2\varepsilon)$-balanced.

\begin{proof}
Pick $m$ to be the smallest integer so that $n \le \varepsilon q^m$. Let $Q = q^m$. Note that $Q \le q n/\varepsilon$.
   Let $\text{RS} : \mathbb F_q^n \to \mathbb F_Q^Q$ be a Reed-Solomon encoding map that maps polynomials of degree $< n$ over $\mathbb F_q$ to their evaluations at all points in the extension field $\mathbb F_Q$. Now concatenate this encoding with the Hadamard encoding that maps $\mathbb F_Q$, viewed as vectors in $\mathbb F_q^m$ under some canonical basis, to $\mathbb F_q^{q^m}$. The resulting concatenated code has block length $N=Q \cdot q^m = Q^2  \le ( q n/\varepsilon)^2$.

   The distance of the concatenated code is at least $(1-\varepsilon)(1-\frac{1}{q}) N$, since the Reed-Solomon code has distance greater than $Q - n \ge (1-\varepsilon) Q$ and the Hadamard code has distance $(1-\frac{1}{q})Q$. The lower bound on the weight of every nonzero codeword follows from distance, while the upper bound comes from the fact that each of the $Q$ symbols in Reed-Solomon code contributes at most $(1-\frac{1}{q})Q$ Hamming weight after encoding by the Hadamard code.
\end{proof}

\subsection{Reduction for MDP($\F_q$)}
\label{subsec:mdp_red}
We now prove \Cref{thm:mdp} by presenting a gap-producing reduction from homogeneous quadratic equations to MDP($\F_q$).

\noindent\textbf{Input.} A parameter $\varepsilon>0$ and a system of homogeneous quadratic equations of the form
\begin{align}
\label{hq}
Q_1(xx^T)=0,\dots , Q_m(xx^T)=0.
\end{align}

\noindent\textbf{Output Subspace.}
Let $\varepsilon = \frac{1}{9q}$, and let $G\in \mathbb{F}_q^{N\times n}$ be the generator matrix of an $\varepsilon$-balanced code of minimum distance $d$. Our output subspace
is defined as
\begin{equation}
V:= \{GXG^T :~ Q_1(X)=0,\dots , Q_m(X)=0,~
X^T=X,~X\in\mathbb{F}_q^{n\times n}\}.
\label{hred}
\end{equation}

Using the construction of an $\varepsilon$-balanced code from \Cref{lem:eps-biased}, we reduce an instance of homogeneous quadratic equations with $n$ variables to an MDP($\F_q$) instance with $N^2=\text{poly}\left(n,\frac{1}{\varepsilon}\right)$ variables. A basis of $V$ can be computed in polynomial time by considering
the basis $\{GXG^T:X\in B\}$, where $B$ is a basis of
$\{X:Q_1(X)=0,\ldots ,Q_m(X)=0,~X^T=X\}$.

\subsection{Analysis}
\noindent\textbf{Completeness.} Let $x \in \mathbb \{0,1\}^{n}$ be a non-zero solution to the system \Cref{hq}. Then $(Gx)(Gx)^T\in V\setminus\{0\}$ and satisfies $\|(Gx)(Gx)^T\|_0 \le (1+\varepsilon)^2 d^2 \leq \left(1+\frac{1}{3q}\right)d^2$.

\noindent\textbf{Soundness.} Suppose there is no non-zero solution to system \Cref{hq}, we argue that any $Y \in V \setminus \{0\}$ has $\|Y\|_0 \ge \left(1+\frac{1}{q}\right) d^2$. Consider any $Y \in V \setminus \{0\}$ and let
$X$ be such that $Y=GXG^T$.

If $X$ has rank at least 2, then $GXG^T$ has rank at least 2 since $G$, being the generator matrix of a code of positive distance, has full column rank. It then follows from \cref{lem:rk2hwfq} that $\|Y\|_0\ge \left(1+\frac{1}{q}\right) d^2$. So it remains to consider the case where $X$ has rank 1. Since $X$ is symmetric, we conclude $X = xx^T$ for some non-zero $x\in \mathbb F_q^n$, which implies that for every $\ell \in [m]$, $Q_\ell(xx^T)=0$,
and thus $x$ is a solution to the system \Cref{hq}, contradicting our assumption.

This yields NP-Hardness of approximating MDP($\F_q$) within a factor of
$\left(1+\frac{1}{q}\right)/\left(1+\frac{1}{3q}\right)= 1+\frac{2}{3q+1}$.
By simple tensoring (\cref{tensoring}), one can increase the gap to any constant, with only a polynomial blow-up on the instance size, and to almost-polynomial gap with a quasi polynomial blow-up in instance size.  This completes the proof of \Cref{thm:mdp}.

By slightly modifying the reduction to utilize a distinguished-coordinate property, we deduce the hardness of NCP (\Cref{thm:ncp}) as a quick corollary of \Cref{thm:mdp}. We put the proof in \Cref{sec:ncp}.

\section{PCP-Free Randomized Reduction for MDP($\R$) and SVP}
\label{bypass:pcp}
The main result of this section is \cref{thm:nsp}.
The starting point of our reduction is hardness of exactly solving a non-homogeneous system of quadratic equations. Throughout this section, we use $H^N_k$ to denote the weight-$k$ slice of the $N$-dimensional hypercube, i.e., $H^N_k:=\{x \in \{0,1\}^N \mid \|x\|_0=k\}$.

\subsection{Rank-1 Testing over $\R$ via Tensor Codes}
Just like over $\F_q$, the key to our gap-producing reduction for MDP($\R$) is a two-sided version of \cref{lem:rk2hw}. There is a nontrivial complication over $\R$ compared to $\F_q$: there is no $n$-dimensional subspace $C \subseteq \R^N$ with distance $n^{\Omega(1)}$ and non-overlapping coefficient $1+\Omega(1)$, and that further admits an encoding map $G : \R^n \to C$ mapping $\{0,1\}^n$ to vectors of Hamming weight at most $\rho \cdot d(C)$ for small $\rho \geq 1$.
Inspired by \cite{ajtai1998shortest,micciancio2001shortest}, we weaken the final requirement above to: \\
there exists a linear projection
$T:\R^N\to \R^n$ (for $n=N^{\Omega(1)}$) such that
\begin{equation}
\label{wlc}
    \text{$\forall~x\in \{0,1\}^n,~\exists~y\in C\cap H_{\rho d(C)}^N$, s.t. $Ty=x \qquad\equiv\qquad T(C\cap H_{\rho\cdot d(C)}^N)\supseteq \{0,1\}^n$}\mper
\end{equation}
Collecting these requirements, we define:

\begin{definition}[Coding Gadget]
    Let $\F$ be a field. For $\rho\geq 1,\alpha\geq 1,n\in\N$, we say a triple $(C,T,k)$ is a $(\rho,\alpha,n)$-coding gadget if for some $d,N\in \N$,
    $C$ is a subspace in $\F^N$, $T\in \F^{n \times N}$, and $k \le \rho \cdot d(C)$, and they satisfy
    \begin{itemize}
    \itemsep=-0.5ex
    \vspace{-1ex}
        \item (Weak Local Density): $T(C \cap H_{k}^N)\supseteq \{0,1\}^n$.
        \item (Non-Overlap): $C$ is $\alpha$-non-overlapping.
    \end{itemize}
\end{definition}
\noindent
The connection between \cref{wlc} and local density, i.e., an abundance of codewords in a small Hamming ball,
is clarified via
the Sauer-Shelah lemma, which states that a sufficient condition to ensure
\cref{wlc} is to have $|C\cap H_{\rho\cdot d(C)}^N| > \sum_{i=0}^n \binom{N}{i}$, in which case $T$ can be taken to be the projection to $n$ {\em shattered} coordinates. Combined with the trivial necessary condition $|C\cap H_{\rho\cdot d(C)}^N| \geq 2^n$, having \cref{wlc} with $n = N^{\Omega(1)}$ is equivalent to a code having exponentially many codewords in a ball of radius
$\rho\cdot d(C)$ (where we call it weak local density since $\rho\geq 1$).
Micciancio~\cite{micciancio2001shortest} proved a probabilistic version of the Sauer-Shelah lemma\footnote{using a matrix with i.i.d. Bernoulli entries}, and so it suffices for our purposes to construct a family of codes that are $\alpha$-non-overlapping and $\rho$-locally-dense where $\rho^2<\alpha$.

In \cref{section:overlap}  we show that the kernel of a random Rademacher matrix has strong non-overlapping properties. Combining this with a local density bound from \cite{FSSZ23} and Micciancio's random linear projection, we obtain
\begin{restatable}{theorem}{codinggadget}(Computing the Coding Gadget over Reals)
\label{coding:gad}
    Fix any $\eps\in (0,1)$. There is a randomized algorithm that on input $n\in \N$,
    runs in time $n^{O(1)}$ and for some integers $d,h,N$ produces matrices
    $R\in \{\pm 1\}^{h\times N},~T\in \{0,1\}^{n\times N}$ and an integer $k$ such that $(\ker(R),T,k)$ is a $(1+\eps,~2-\eps,~n)$-coding gadget
    with probability $1-o(1)$.
\end{restatable}
\begin{remark}
We remark that the bounded integer random entries model seems influential
in our successful construction of a coding gadget. Gaussian random matrices will not work for instance.
\end{remark}

In the following, we first show how to perform the reduction given the gadget from \Cref{coding:gad}, and defer the proof of \Cref{coding:gad} to \Cref{random:properties}.

\subsection{Reduction}
\label{subsec:mdpred}
We next present our basic reduction that generates a constant multiplicative gap.

We will use the following randomized polynomial time reduction from non-homogeneous quadratic equations to MDP($\R$).

\noindent\textbf{Input.} A parameter $\eps \in (0, 1)$ and a system of $m$ quadratic equations of the form
\begin{align}
\label{hq2}
  Q_1(xx^T) = b_1, \ldots,Q_m(xx^T) = b_m.
\end{align}

\noindent \textbf{Output Subspace.}
Let $(C,T,k)$ be a $(\rho,\alpha,n)$-coding gadget for $\rho= 1+\eps,~\alpha= 2-\eps$ as computed in \cref{coding:gad}.
The linear subspace $V$ is the set of tuples $(Y,z) \in \R^{N \times N}\times \R$ satisfying the following system of homogeneous linear equations in $Y, z$.
\begin{align}
     &Y\in C\otimes C \nonumber \\
     &Y = Y^T \nonumber \\
     &z = \sum_{i \in [n]} Y[i, i] / k \nonumber \\
    &\!\!\!\!\!\!Q_\ell(TYT^T) = z \cdot b_{\ell}  \quad \forall \ell \in [m]
    \label{l0:newreduction}
\end{align}
where \Cref{l0:newreduction} refers to the entire system rather than just the last line.

Since $C=\ker(R)$ and $T\in \{0,1\}^{N\times n}$, it is easily verified that the above subspace can be written as $\ker(M)$
for an $N^{\Theta(1)}\times N^{\Theta(1)}$ matrix $M$ with integer entries of magnitude $O(N^2)$.

\subsection{Analysis}
\label{sec:svp_analysis}

\noindent\textbf{Completeness.} Let $x \in \{ 0, 1 \}^{n}$ be a solution to the system \Cref{hq2}.
By (Weak Local Density) of the coding gadget, there exists $y \in \{ 0, 1 \}^N \cap \ker(R)$ with $\norm{0}{y} = k$ such that $Ty = x$. Clearly $Y = yy^T$, $z = 1$ is a nonzero Boolean
solution of \Cref{l0:newreduction} which satisfies $\norm{0}{(Y,z)} = k^2+1\leq \rho^2\cdot d(C)^2 +1$, so we take $s = \rho^2\cdot d(C)^2 +1$. \bigskip

\noindent\textbf{Soundness.} We proceed via the contrapositive. Let $(Y,z) \in V$ be nonzero with $\norm{0}{(Y,z)} < \alpha\cdot d(C)^2$. We will show that the system of quadratic equations admits a solution over reals.

If $Y=0,z\neq 0$, the final constraints of \Cref{l0:newreduction} imply that $b= 0$
and so $0$ is a solution to the system of quadratic equations.

If $Y$ has rank at least 2, then since $C$ is $\alpha$-non-overlapping,  \cref{lem:rk2hw} implies $\norm{0}{Y} \geq \alpha\cdot d(C)^2$, which contradicts our assumption.

Thus $Y$ must have rank $1$. Since $Y$ is symmetric, we conclude $Y = yy^T$ for some nonzero $y \in \R^N$.
The final constraints of \Cref{l0:newreduction} imply that for any $\ell \in [m]$,
\[
    Q_\ell(Ty(Ty)^T) = z\cdot b_\ell\mper
\]
Moreover the constraint $z = \sum_{i} Y[i, i] / k= \norm{2}{y}^2/k$ implies that $z > 0$. Thus $Ty / \sqrt{z}$ is a solution to the system of quadratic equations.

For $k$ sufficiently large, we obtain a gap of $\alpha/\rho^2-o(1) \geq 2-3\eps$.
\cref{thm:nsp} then follows from applying \cref{tensoring}, where we note that the Booleanity property of the completeness solution is preserved
under tensoring.

Our strong completeness and soundness conditions imply hardness of SVP (\cref{thm:svp}) as an immediate corollary (while expanding the range of values of $p$ for which hardness was known):
\begin{proof}[Proof of \cref{thm:svp}.]
Given an instance $(V, s)$ of \cref{thm:nsp}, define the lattice $L := V \cap \Z^n$. In the (YES) case, there exists $x \in L \cap H^N_s$, and so
$\| x \|_p^p = s$. In the (NO) case, if $x \in L \setminus \{ 0 \}$ satisfies $\norm{0}{x} \geq \gamma \cdot s$ for some gap parameter $\gamma$, since every nonzero coordinate has absolute value $\geq 1$, we must have $\| x \|_p^p \geq \gamma\cdot s$.
\end{proof}

\begin{remark}
    We remark that the proof of \cref{thm:nsp} works over any field $\F$ provided one can construct a $(\rho,\alpha,n)$-coding gadget over $\F$
    for $\rho^2<\alpha$. We believe our approach can be extended to work for fields whose size grows with the input.
\end{remark}

\begin{remark}
    We can replace the variable $z$ in \Cref{subsec:mdpred} with the constant 1 to obtain CVP hardness for an explicit constant factor. The analysis in \Cref{sec:svp_analysis} still goes through. Note that in this way, we change the linear equations from homogeneous to non-homogeneous, and implicitly rule out the all-zeros solution.

    For larger gap, the same approach applies to the tensoring
    of the reduction in \Cref{subsec:mdpred}, wherein there is a
    distinguished variable $z'$ that we may set to $1$ while still
    preserving at least one solution in the completeness case (namely the tensoring of the completeness solution defined in  \Cref{sec:svp_analysis}, is a feasible solution of small sparsity for the non-homogenized tensored instance ).
    This allows $2^{\log^{1-\eps}n}$ gap assuming \NP is not contained in randomized quasipolynomial time.

    We remark that for CVP, we don't achieve the state of the art (deterministic) hardness factor of $n^{\Omega(1/\log \log n)}$ that is due to \cite{DKRS00}.
\end{remark}

\newcommand{\bases}{\mathrm{Bases}}
\section{Weak Local Density and Non-Overlap for Rademacher Kernel}
\label{random:properties}
Throughout this section, let $R_{h,N}$ denote an $h\times N$ matrix with i.i.d. $\pm 1$ (Rademacher) random entries. Wherever it is clear from context, we will drop the subscript and use $R$.

In this section we prove \cref{coding:gad}. The algorithm simply outputs $\ker(R),T$, where $T$ is a matrix with i.i.d. random Bernoulli entries of appropriately chosen bias.
The primary technical work to be done is proving that for appropriate values of $h,N$, $\ker(R)$ satisfies (Non-Overlap), moreover with a surprisingly strong coefficient of $2-\eps$. We do this in \cref{section:overlap}, with the key lemma (\Cref{opair:SBP}) on the {\em small ball probability} proved in \Cref{section:smallball}.
In \cref{summary} we stitch together the remaining properties of the coding gadget using a sharp phase transition result for the minimum distance of $\ker(R)$ due to \cite{FSSZ23} and a probabilistic Sauer-Shelah lemma proved in \cite{micciancio2001shortest}.

\subsection{Preliminaries and Overview}
In this section, we collect the necessary preliminaries for proving (Non-Overlap). En route we give a gentle overview of the proof, and also introduce the aforementioned off-the-shelf results.
We begin with a discussion of the minimum distance, which serves as a warmup for (Non-Overlap).

\medskip\noindent\textbf{Phase transition for the Boolean Slice.}
To provide intuition for the technical aspects of this section, we discuss how
$|\ker(R)\cap H_d^N|$ undergoes a sharp phase transition ($d$ being the parameter varying as a function of $h,N$), jumping from zero to exponential.

If $\xi_1,\ldots,\xi_d$ are i.i.d. Rademacher ($\pm 1$) random variables, we have by Stirling's approximation that $\Pr{\xi_1+\ldots +\xi_d = 0} \asymp 1/\sqrt{d}$.
Then for any $u\in H_d^N$, ~$\Pr{Ru=0}= \Theta(1/\sqrt{d})^{h}$.
On the other hand, $|H_d^N|={N\choose d}$.
Thus, we have $\Ex{|\ker(R)\cap H_d^N|}= {N\choose d}\cdot \Theta(1/\sqrt{d})^{h}$. It follows that for any fixed $\eps>0$, $\Ex{|\ker(R)\cap H_{(1-\eps)d}^N|}=N^{-\Omega(d)}$
and $\Ex{|\ker(R)\cap H_{(1+\eps)d}^N|}=N^{\Omega(d)}$ when
$h= \floor{\log_{\sqrt{d}}{N \choose d}}\sim d\log_{\sqrt{d}}(N/d)$. It is not difficult to prove a high probability version of this  statement. Indeed by union bound, $|\ker(R)\cap H_{(1-\eps)d}^N|=0$
w.h.p. Combining Chebyshev's inequality with a careful estimate on the variance, one can show that w.h.p. $|\ker(R)\cap H_{(1+\eps)d}^N|=N^{\Omega(d)}$:
\begin{lemma}[Boolean Weak Local Density of Rademacher Kernel~\cite{FSSZ23}]\footnote{see proof of Theorem 1.2 (2) in \cite{FSSZ23}}~\\
\label{scnd:mmnt}
Fix any  $\eps,\delta\in (0,1)$. For any $h\in \N$ sufficiently large, let $d=\ceil{\delta \cdot h}$. If $N$ is the largest integer satisfying $h\geq d\log_{\sqrt{d}}(N/d)$, then $|\ker(R_{h,N})\cap H^N_{(1+\eps)d}| \geq (N/2)^{\eps d}$ with probability $1-o(1)$.
\end{lemma}
\noindent For the rest of this overview, we fix $h$:
\begin{equation}
\label{overview:h}
    h = \floor{d\log_{\sqrt{d}}(N/d)} \ .
\end{equation}

\medskip\noindent\textbf{Phase transition for Sparse Vectors with Real Entries.}
Define $\dsp{d}{N}$ to be the set of unit vectors in $\R^N$ with sparsity $\leq d$:
\[
    \dsp{d}{N} := \{ u \in \R^{N} : \norm{2}{u}=1,~\norm{0}{u}\leq d \}
\]
$\dsp{d}{N}$ contains vectors such as $u=e_1+e_2$ or $u=e_1+e_2 + (e_3+\dots +e_d)/d$ (where $e_i$ denotes the $i$-th elementary vector) for which $Ru$ is close to the origin with a relatively high probability of $1/2^h$.

Despite this, it has been established in \cite{FSSZ23} that the avoidance threshold
for the Boolean $d$-slice is roughly equal to the avoidance threshold for all $d$-sparse real vectors:
\begin{theorem}[Minimum Distance for Random Rademacher Kernel~\cite{FSSZ23}]
\label{distance}
Fix any  $\eps,\delta\in (0,1)$. For any $h\in \N$ sufficiently large, let $d=\ceil{\delta \cdot h}$. If $N$ is the largest integer satisfying $h\geq d\log_{\sqrt{d}}(N/d)$, then $\dsp{(1-\eps)d}{N}\cap \ker(R_{h,N}) = \emptyset$ with probability $1-o(1)$.
\end{theorem}
The proof in \cite{FSSZ23} uses precise estimates from inverse Littlewood-Offord theory
that count the number of coefficient vectors $a\in \F_p^d$ such that the
small ball probability of $\sum_i a_i \xi_i$ is in a particular range. There is then a
careful tradeoff between the union bound size and the small ball probability.

We next sketch a simple proof of a weaker result that matches \cref{distance} when $N \geq d^{1/\eps}$. This proof
will serve as a prototype for our proof of the non-overlapping property for $\ker(R)$.
The key idea is that the RIP property rules out all vectors that have (relatively) high probability of lying in a small ball around the origin. The remaining vectors have low probability of being in a small-ball around the origin, and we may take a union bound over a sufficiently fine net.

\medskip\noindent\textbf{Ruling Out Compressible Vectors.}
We begin with the observation that $R$ satisfies the restricted isometry property (RIP) for restrictions of size $\succeq d/\log d$.
In particular this means that all vectors in $\ker(R)$ have sparsity
$\succeq d/\log d$. RIP implies something stronger, namely that all vectors in $\ker(R)$ have $\ell_1/\ell_2$ ratio $\succeq d/\log d$:
\begin{theorem}[Width Property for Rademacher Kernel~\cite{MP03,MPT07,BDDW08}]~\\
\label{rademacher:width}
    For any $h,N\in \N$, with probability at least $1-e^{c_0 h}$,
    every $u\in \ker(R_{h,N})$ satisfies
    \[
        \norm{1}{u}\geq c_0\sqrt{\frac{1+\log (N/h)}{h}}\cdot\norm{2}{u},
    \]
    where $c_0>0$ is a universal constant.
\end{theorem}
The above result rules out vectors of sparsity $d$ that
have $1-\Omega(1/\sqrt{\log d})$ fraction of their $\ell_2$ mass concentrated in $\preceq d/\log d$ entries
(see \cref{no:compressible}).
Such vectors are referred to as \emph{compressible} vectors in the random matrix literature and
are defined as:
\[
    \comp^d_{\rho,\delta} := \Big\{ u\in \dsp{d}{N} ~|~ \exists \bT \subseteq \supp(u), |\bT| \leq \delta  d, \mbox{ s.t. }
    \| u_{\supp(u)\setminus\bT} \|_2 \leq \rho \Big\}
\]
i.e., the set of vectors of sparsity $d$ that are $\rho$-close to a
$\delta d$-sparse vector. As stated above, \cref{rademacher:width} implies that $\ker(R)$ avoids compressible vectors with parameters $\rho,\delta$
inverse logarithmic in $d$:
\begin{corollary}
\label{no:compressible}
    Let $\rho,\delta>0,d\in\N$ be such that
    $\delta \cdot d< (c_0^2/4) d/\log d$ and $\rho < (c_0/2)/\sqrt{\log d}$,
    where $c_0>0$ is the same constant as in \cref{rademacher:width}. Let $h,N\in \N$.
    Then with probability at least $1-e^{-c_0 h}$,
    $\ker(R_{h,N})\cap \comp^d_{\rho,\delta}=\emptyset$.
\end{corollary}
\begin{proof}
    Consider any $u\in\comp^d_{\rho,\delta}$. Let $T$ be a subset of size $\delta\cdot d$ such that $\norm{2}{u_{\supp(u)\setminus T}}\leq \rho$.
    By the Cauchy-Schwarz inequality, we have
    \[
        \norm{1}{u}= \norm{1}{u_T}+\norm{1}{u_{\supp(u)\setminus T}} \leq
        \sqrt{\delta \cdot d}\norm{2}{u_T} + \sqrt{d}\norm{2}{u_{\supp(u)\setminus T}}
        < c_0 \sqrt{d/\log d}.
    \]
    By \cref{rademacher:width}, with probability $1-e^{-c_0h}$, no such vector $u$ can lie in $\ker(R)$.
\end{proof}
Let $d':= (1-\eps)d$. \cref{no:compressible} implies $\ker(R)\cap\comp^{d'}_{\rho,\delta}=\emptyset$ for $\rho^2,\delta\preceq 1/\log d$. It remains to rule out \emph{incompressible} $d'$-sparse vectors in $\ker(R)$ where
\[
    \incomp^{d'}_{\rho,\delta} := \dsp{d'}{N}\setminus \comp^{d'}_{\rho,\delta} \ .
\]
For $\rho^2,\delta\preceq 1/\log d$, incompressible vectors $u\in \incomp^{d'}_{\rho,\delta}$ have $\succeq d/\log d$ entries of magnitude
$\succeq 1/\sqrt{d\log d}$.
By the classical Littlewood-Offord inequality~\cite{LO39}, any such vector $u$ has low small-ball probability, i.e.,  $\Pr{|\sum_i u_i\xi_i|\leq \log^{O(1)}d/\sqrt{d}}\leq \log^{O(1)}d/\sqrt{d}$.
It follows that $\Pr{\norm{\infty}{Ru}\leq \log^{O(1)}d/\sqrt{d}}\leq \log^{O(h)}d/d^{h/2}$. To rule out all such vectors from being in $\ker(R)$, it suffices to take a union bound of the bad event $\norm{\infty}{Ru}\leq \log^{O(1)}d/\sqrt{d}$ over a $1/d^2$-net (in the $\ell_2$ metric) of $\dsp{d'}{N}$.
Such a net has size $d^{O(d)}\cdot {N\choose d'}$, and thus $\log^{O(h)}d/d^{h/2}$
is a sufficiently small probability for the union bound to succeed, provided $N\geq d^{C/\eps}$ for a sufficiently large constant $C$.
\bigskip

\noindent
\textbf{Nets.}
It is convenient to briefly define and discuss $t$-nets here.
\begin{definition}
    Let $T$ be a metric space with distance $D$ and let $E \subset T$. $\cN \subseteq T$ is called an $t$-net of $E$ if for every $x \in E$ there exists $y \in \cN$ such that $D(x, y) \leq t$.
\end{definition}
We use a simple volumetric estimate on the net size of the sphere. We will also need a means of passing from a net of a set $E$
to a net that lies within $E$. Both facts are standard and their proofs can be found for instance in \cite{Vershynin17}.
\begin{fact}
	There is a $t$-net (in $\ell_2$ norm) within $\SSS^{d-1}$, of size at most $(\frac{6}{t})^d$.
\label{volume:net}
\end{fact}

\begin{fact}
\label{net:within}
 Let $T$ be a metric space and let $E \subset T$.
 Let $\cN \subset T$ be a $t$-net of the set $E$.
 Then there exists a $(2t)$-net $\cN'$ of $E$
 whose cardinality does not exceed that of $\cN$,
 and such that $\cN' \subset E$.
\end{fact}

\smallskip\noindent
\textbf{Restricted Maximum Singular Value.}
Returning to the discussion of minimum distance of a random Rademacher kernel, we bound the error incurred in $\norm{\infty}{Ru}$ in passing from $u\in \dsp{d'}{N}$ to the net point closest to it, in terms of the maximum singular value of $R$ restricted to the columns in $\supp(u)$.
Since we require such an estimate for all size-$k$ submatrices, we will use an estimate on
$\sigma^{d'}_{max}(R) := \max_{|T| \leq d', |T| \subseteq [N]} \sigma_{max} (R_T)$.
We have
\begin{proposition}[Restricted Maximum Singular Value~\cite{Kashin97}; see also \cite{KT07}]~\\
\label{restr:msv}
    There is a constant $C_0>1$, such that for any $d'\leq h$, we have $\sigma_{max}^{d'}(R_{h,N}) \leq C_0 \sqrt{h}(1+\log (N/h))$ with probability $1-o(1)$.
\end{proposition}

\subsection{Non-Overlap for Random Rademacher Kernel}
\label{section:overlap}
In this section we establish the somewhat surprising fact that any pair of linearly independent vectors in $\ker(R)$ of near-minimum sparsity overlap in a negligible fraction of their support. We deduce this as a consequence of the fact that the second order Hamming weight of $\ker(R)$ is nearly
twice the minimum distance:
\begin{definition}[Higher Order Hamming Weight]\label{def:overlap}~\\
    For $\ell\in \N$, the order-$\ell$ Hamming weight of a subspace $C\subseteq \F^N$ (denoted as $d_\ell(C)$) is the smallest joint support size
    $|\sigma^{u_1}\cup \dots \cup\sigma^{u_\ell}|$
    for any collection of linearly independent vectors $u_1,\dots ,u_\ell\in C$.
\end{definition}
\noindent
We will prove $d_2(\ker(R))\geq 2(1-\eps)d$. We first sketch the proof, collecting necessary preliminaries enroute.
\medskip

\noindent
\textbf{Non-Overlap for the Boolean Slice.}
Similar to the minimum distance, the case of Boolean vectors is instructive to appreciate the quantitative aspects of the non-overlapping property.

Let $d_2:=2(1-\eps)d$. No two linearly independent vectors $u,v\in H_{\leq d_2}^N\cap\ker(R)$ can overlap in all but $\preceq d/\log d$ coordinates, since in that case $u-v\in\ker(R)$ has sparsity $\preceq d/\log d$ which would contradict \cref{no:compressible}.

So, it remains to consider the case where the symmetric difference of $\supp(u),\supp(v)$ is of size $\succeq d/\log d$. In the extreme case where they have disjoint support, by independence we have
\[
    \textstyle \Pr{\sum_i \xi_i \cdot u(i)=\sum_i \xi_i \cdot v(i) = 0}=\Pr{\sum_{i\in\supp(u)} \xi_i =0}\cdot \Pr{\sum_{i\in\supp(v)} \xi_i = 0}
    \preceq 1/d.
\]
Thus $\Pr{Ru=Rv=0}\leq O(1/d)^h$. A similar bound is possible for the
more general case where the symmetric difference $\supp(u) \triangle \supp(v)$ has size $\succeq d/\log d$. Indeed we have
\begin{align*}
    &\Pr{\textstyle\sum_i \xi_i \cdot u(i)=\sum_i \xi_i \cdot v(i) = 0}\\
    &=
    \sum_{t=-d_2}^{d_2}\mathbb{P}\big[{\sum_{i\in \supp(u)\setminus \supp(v)} \xi_i =-t ~~~~\wedge \sum_{i\in \supp(v)\setminus \supp(u)} \xi_i =-t}\big]
    \cdot \mathbb{P}\big[\sum_{i\in \supp(u)\cap \supp(v)} \xi_i=t\big]\\
    &=
    \sum_{t=-d_2}^{d_2}\mathbb{P}\big[\sum_{i\in \supp(u)\setminus \supp(v)} \xi_i =-t\big] \cdot \mathbb{P}\big[\sum_{i\in \supp(v)\setminus \supp(u)} \xi_i =-t\big]
    \cdot \mathbb{P}\big[\sum_{i\in \supp(u)\cap \supp(v)} \xi_i=t \big]\\
    &\preceq
    \sum_{t=-d_2}^{d_2}\frac{1}{\sqrt{|\supp(u)\setminus \supp(v)|}\cdot\sqrt{|\supp(v)\setminus \supp(u)|}}
    \cdot \mathbb{P}\big[\sum_{i\in \supp(u)\cap \supp(v)} \xi_i=t \big]
    \quad \preceq \frac{\log d}{d}.
\end{align*}
where the second to last step can be shown using (say) Stirling's approximation
for binomial coefficients. So we have $\Pr{u,v\in\ker(R)}
\leq \log^{O(h)}d/d^h$.
There are at most ${N\choose d_2}2^{O(d)}$ choices of pairs $u,v\in H_{\leq d_2}^N$ such that $|\supp(u)\cup \supp(v)|=d_2$. Thus we may take union bound
and obtain that w.h.p., any pair of boolean vectors $u,v\in \ker(R)$
must satisfy $|\supp(u)\cup \supp(v)|> d_2$.

\medskip\noindent\textbf{Non-Overlap for Vectors with Real Entries.}
We refine the approach for the Boolean slice so that it applies to vectors with real entries.
This time we will begin with the observation that the span of two linearly independent vectors $u,v\in S_{d_2}^N\cap\ker(R)$
cannot contain a compressible vector $w\in \comp^{d_2}_{\rho,\delta}$ for $\rho^2,\delta\preceq 1/\log d$, since that would contradict \cref{no:compressible}.

It then remains to consider the case of linearly independent pairs
$u,v\in S_{d_2}^N\cap\ker(R)$ that do not contain a compressible vector in their span---henceforth referred to as incompressible pairs. We will organize such incompressible pairs according to the subspace they span and identify them with an orthonormal basis. For $d_2\leq N$ we define
\begin{align*}
    \bases_{\rho,\delta}^{d_2} &:= \big\{ (u_1,u_2)\in \SSS^{N-1}\times \SSS^{N-1} ~~\big|~~ u_1,u_2\text{ orthonormal},\\
    &\qquad\quad|\supp(u_1)\cup \supp(u_2)| \leq d_2, ~ \mathrm{Span}(\{u_1,u_2\})\cap \comp^{d_2}_{\rho,\delta} = \emptyset \big\}\mper
\end{align*}

We use a powerful anticoncentration result of \cite{rudelson2009smallest} to
deduce that the small ball probability of an incompressible pair behaves
similarly to the independent case:
\begin{restatable}{corollary}{opairSBP}(Joint Small Ball Probability for Basis Elements of an Incompressible Subspace)~\\
\label{opair:SBP}
     Fix any $\rho,\delta\in (0,1)$. For any $d\in \N$ sufficiently large, any $h,N\in \N$,
     and any $(u_1,u_2)\in \bases^{2d}_{\rho,\delta}$,
     \[
         \Pr{\norm{\infty}{(R_{h,N})u_1},\norm{\infty}{(R_{h,N})u_2} \leq 1/\sqrt{\delta d} } \leq C_1^h (\rho^4\delta d)^{-h}\mper
     \]
     where $C_1>1$ is a universal constant.
\end{restatable}
We defer the proof of this corollary to the next section as it requires some new notions.

Equipped with our joint small ball estimate, we may then take a union bound
over a sufficiently fine net of $\bases_{\rho,\delta}^{d_2}$
to establish the bounded overlap property. We have assembled all ingredients required to give the proof:

\begin{theorem}[Non-Overlap for a Random Rademacher Kernel]~\\
\label{overlap}
There is a universal constant $c_2>0$ such that the following holds. Fix any $\eps\in (0,1),~\delta\in (0,c_2\eps]$,
let $h\in \N$ be sufficiently large, and let $d:=\ceil{\delta\cdot h}$.
If $N$ is the largest integer satisfying $h \geq d \log_{\sqrt{d}} (N/d)$, then with probability $1-o(1)$,
$d_2(\ker(R_{h,N}))\geq 2(1-\eps)d$.
\end{theorem}

\begin{proof}
Let $d_2:= 2(1-\eps)d$ and let $R$ denote $R_{h,N}$.
By \cref{no:compressible}, there is a universal constant $c_0>0$
such that $\comp^{d_2}_{\rho_0,\delta_0}\cap \ker(R)=\emptyset$ with probability $1-o(1)$, where $\rho_0:= c_0/\sqrt{\log d},~\delta_0:=c_0/\log d$.

For any linearly independent vectors $\tilu_1,\tilu_2\in \ker(R)$ satisfying $|\textstyle \supp(\tilu_1)\cup \supp(\tilu_2)|\leq d_2$,
it holds that any orthonormal basis $u_1,u_2$ of $\spn\{\tilu_1,\tilu_2\}$ satisfies $u_1,u_2\in \ker(R)$ and $|\textstyle \supp(u_1)\cup \supp(u_2)|\leq d_2$. Thus the claim follows if we show that with probability $1-o(1)$, for any orthonormal pair $u_1,u_2\in\dsp{d_2}{N}$ one of $u_1,u_2$
does not lie inside $\ker(R)$.
Since we showed above that $\ker(R)$ avoids compressible vectors with high probability, we may assume that $\spn\{u_1,u_2\}$ is incompressible.
So we need only show that $\bases^{d_2}_{\rho_0,\delta_0}\cap (\ker(R)\times \ker(R)) \neq\emptyset$ with probability $o(1)$, which we do by combining \cref{opair:SBP} with a union bound over a sufficiently fine net.

Let $\calO\subseteq \bases^{d_2}_{\rho_0,\delta_0}$ be a minimum size $1/d^{2}$-net of $\bases$ according to the norm $\|(u_1,u_2)\| := \max\{\|u_1\|_2,\|u_2\|_2\}$. We next show that for some constant $c_0>0$,
\begin{equation}
\label{net:SBP}
\Pr{\min_{(u_1,u_2) \in \calO} \max_{i\in [2]}\| Ru_i \|_\infty \leq \sqrt{\log d/(c_0d)}}=o(1).
\end{equation}
We proceed by a small-ball probability estimate, followed by a union bound. For any fixed $(u_1,u_2)\in \bases^{d_2}_{\rho_0,\delta_0}$, we have by
\cref{opair:SBP} (substituting $\rho\leftarrow \rho_0,\delta\leftarrow \delta_0$ and noting $d_2= \Theta(d)$) that for some universal constants $c>0,C>1$,
\[
    \Pr{\norm{\infty}{Ru_1},\norm{\infty}{Ru_2}  \leq c\sqrt{\log d/(d)}} \leq (C\log^{3} d/d)^{h}
    = (d/N)^{2d-o(d)}\mper
\]
As for net size, observe that $\bases^{d_2}_{\rho_0,\delta_0}$ has a $1/(2d^{2})$-net of size $d^{O(d)}\cdot N^{d_2}$ since the $d_2$-dimensional sphere has a $1/(2d^2)$-net of size $d^{O(d)}$ (\cref{volume:net}), and any basis in $\bases^{d_2}_{\rho_0,\delta_0}$ can be generated by choosing a subset of size $d_2$ and then choosing two vectors in the sphere supported on those coordinates.
We then apply \cref{net:within} to obtain a $1/d^{2}$-net for $\bases^{d_2}_{\rho_0,\delta_0}$ that is also a subset. So we have $|\calO|= d_2^{O(d_2)}\cdot N^{d_2}$. Thus the probability that $\min_{(u_1,u_2) \in \calO} \max_{i\in [2]}\| Ru_i \|_\infty \leq \sqrt{\log d_2/(c_0k)}$  is at most $d^{O(d)}/N^{\eps d}
\leq 1/N^{d(\eps -O(\delta))}$.

It remains to argue that the minimum over the net is a good lower bound on the true minimum w.h.p.
To this end, it suffices to show that w.h.p.,
\begin{equation}
\label{net:error}
    \inf_{(\tilu_1,\tilu_2) \in \bases^{d_2}_{\rho_0,\delta_0}} \max_{i\in [2]}\| R\tilu_i \|_\infty \geq \min_{(u_1,u_2) \in \calO} \max_{i\in [2]}\| Ru_i \|_\infty ~-~
    O(1/d^{1.4})
\end{equation}
This follows quite easily from estimates on the restricted maximum singular value. Indeed we have
\begin{align}
&\inf_{(\tilu_1,\tilu_2) \in \bases^{d_2}_{\rho_0,\delta_0}} \max_{i\in [2]}\| R\tilu_i \|_\infty  \nonumber\\
&\mathop{\geq}\limits_{\text{(triangle inequality)}}~~
\min_{(u_1,u_2) \in \calO} \max_{i\in [2]}\| Ru_i \|_\infty   ~-~
\max_{\substack{(u'_1,u'_2)\in \bases^{d_2}_{\rho_0,\delta_0},~s.t.\\ \|u'_1\|_2,\|u'_2\|_2 \leq 1/d^{2}}} ~\max_{i\in [2]}\| R u'_i \|_\infty  \nonumber\\
&\mathop{\geq}\limits_{\text{(Cauchy-Schwarz)}}~~
\min_{(u_1,u_2) \in \calO} \max_{i\in [2]}\| Ru_i \|_\infty   ~-~  \sigma_{max}^{d_2}(R)/d^{2} \nonumber \\
&\mathop{\geq}\limits_{(\cref{restr:msv})}~~
\min_{(u_1,u_2) \in \calO} \max_{i\in [2]}\| Ru_i \|_\infty ~-~  O\big(\sqrt{h}(1+\log (N/h))/d^{2}\big) \quad \text{with probability } 1-o(1)
\nonumber\\
&\mathop{\geq}\limits_{\text{(since $N\leq h^{1+1/(2\delta)}$)}}~~
\min_{(u_1,u_2) \in \calO} \max_{i\in [2]}\| Ru_i \|_\infty ~-~
O(1/d^{1.4})
\end{align}
Combining \cref{net:SBP} with \cref{net:error} implies that $\bases^{d_2}_{\rho_0,\delta_0}\cap (\ker(R)\times \ker(R)) \neq\emptyset$ with probability $o(1)$, as desired.
\end{proof}

\begin{remark}
    Our proof easily extends to higher order Hamming weights due to the generality of the multidimensional small ball estimate of \cite{rudelson2009smallest}. It can be shown that $d_\ell(\ker(R)) \geq \ell(1-\eps)d$ for any $\ell\leq c\eps \log_{\sqrt{d}} (N/d)$ where $c>0$
    is some universal constant.
\end{remark}

\begin{remark}
    The restriction $\delta\in (0,c_2\eps]$ in \cref{overlap} can be relaxed significantly by refining the above argument
    following the approach in \cite{RV08,rudelson2009smallest}.
The incompressible vectors  can be further partitioned into level sets of small-ball probability. The level sets can be shown to have attenuating size, using a characterization in \cite{rudelson2009smallest} of the small ball probability in terms of arithmetic structure.
    We choose not to pursue such a refinement in this work, in the interest of clarity.
\end{remark}

\subsection{Small Ball Probability of Incompressible Pairs.}
\label{section:smallball}
It remains to establish \cref{opair:SBP}, for which we rely heavily on
a multidimensional anticoncentration result of \cite{rudelson2009smallest}.
In this section we borrow from the exposition, ideas and notions in \cite{rudelson2009smallest, RV08}.

Following \cite{rudelson2009smallest}, the \emph{essential least common denominator} of a vector $a\in \R^N$, (henceforth abbreviated to LCD), is defined as
\[
    \LCD_{\al,\gamma}(a)
    := \inf \Big\{ \theta>0 ~|~ \dist(\theta\cdot a, \Z^N) < \min(\gamma\|\theta\cdot a\|_2,\al) \Big\}.
\]
It is the minimum scaling of $a$ that is $\alpha$-close to a non-trivial integer point, where a non-trivial integer point is one that lies in a
cone around $a$, enforced by the $\gamma\|\theta\cdot a\|_2$ term in the definition.
and can be thought of as a measure of arithmetic structure.
E.g.
\begin{align*}
    \LCD_{1/10,\sqrt{d}/10}(1,\dots ,1) &\asymp 1,\\
    \LCD_{1/10,\sqrt{d}/10}(1+1/d,1+2/d\dots ,2) &\asymp~ d
\end{align*}
One can even make the LCD polynomially larger by considering polynomial progressions.

Let $E\subset \R^{N}$ be a subspace.
We define
\[
    \LCD_{\al,\gamma}(E)
    := \inf_{a\in\SSS(E)} \LCD_{\al,\gamma}(a).
\]
where $\SSS(E)$ denotes the euclidean sphere restricted to the subspace $E$.

The following theorem which connects multidimensional small ball probability of a signed sum of vectors to the LCD of their Rowspace,
is the main workhorse of our proof of (Non-Overlap).
\begin{theorem}[Two-dimensional Small ball probability, Theorem 3.3 of \cite{rudelson2009smallest}]~\\
\label{SBP}
  Consider a pair of orthonormal vectors $u_1,u_2\in \SSS^{N-1}$. For each $j\in [N]$, let $v^j$ denote the 2-dimensional vector $(u_1(j),u_2(j))$.
  Let $\xi_1, \ldots, \xi_N$ be i.i.d. Rademacher random variables and let $\Xi := \sum_{j=1}^N \xi_j\cdot v^j $ be
  a sum of randomly signed vectors. Then for any $\al > 0$, $\gamma \in (0,1)$, and $t \geq \sqrt{2}/\LCD_{\al,\gamma}(\spn\{u_1,u_2\})$,
  we have
  \[
      \Pr{\norm{2}{\Xi}\leq t\sqrt{2}}
      \leq \frac{Ct^2}{\gamma^2}  + C e^{-\al^2}.
  \]
  where $C$ is a universal constant.
\end{theorem}
\begin{remark}
    Above, we specialized their theorem to the case of orthonormal vectors, and to two-dimensional randomly signed sums. For a reader interested in the details of the
    specialization, see the proof of Theorem 4.2 in \cite{rudelson2009smallest}.
\end{remark}

We require a lower bound on the LCD of vectors that are incompressible.
\begin{lemma}[LCD of incompressible vectors \cite{rudelson2009smallest}]~\\
\label{incomp:LCD}
  Consider any $\rho,\delta \in (0,1),~d\in\N$, and any $a \in \incomp^d_{\rho,\delta}$.
  Then, for every  $\gamma \in (0,\,\,\rho^2\sqrt{\delta}/2)$ and every $\al > 0$, one has $\LCD_{\al,\gamma}(a) > \sqrt{\delta d/2}$.
\end{lemma}

We are finally ready to use \cref{SBP} to derive an estimate on the joint small ball probability of an orthogonal basis whose span does not contain
compressible vectors.
\opairSBP*
\begin{proof}
    Consider any $h,N\in \N$ and let $R:=R_{h,N}$.
    Set $\gamma_0:= \rho^2\sqrt{\delta}/3$.
    \cref{incomp:LCD} gives a lower bound of $\sqrt{\delta d}$, on the LCD of vectors in $\incomp^{2d}_{\rho,\delta}$.
    Indeed for any $\al>0$ we have
    \begin{align*}
        \LCD_{\al,\gamma_0}(\spn\{u_1,u_2\}) = \inf_{a\in\SSS^{N-1}\cap \spn\{u_1,u_2\}} \LCD_{\al,\gamma_0}(a)
        \geq \inf_{a\in\incomp^{2d}_{\rho,\delta}} \LCD_{\al,\gamma_0}(a)
        \geq \sqrt{\delta d}\mcom
    \end{align*}
    where we used the fact that $\SSS^{N-1}\cap \spn\{u_1,u_2\}\subseteq \incomp^{2d}_{\rho,\delta}$ by definition of $\bases_{\rho,\delta}$.

    We then apply \cref{SBP} with $\al\leftarrow \log d,~\gamma\leftarrow \gamma_0,~t\leftarrow 1/\sqrt{\delta d}$,
    to obtain that
    \[
         \Pr{\mysmalldot{\xi}{u_1}^2/2+\mysmalldot{\xi}{u_2}^2/2 \leq 1/(\delta d)}\leq  C\frac{\delta d}{2\gamma_0^2} +
         C e^{-\log^2 d}
         = O(1/(\rho^4\delta d))\mper
    \]
    Since $\max_{i\in [\ell]}\{|\mysmalldot{\xi}{u_i}|\} \geq \big(\mysmalldot{\xi}{u_1}^2/2+\mysmalldot{\xi}{u_2}^2/2 \big)^{1/2}$, we obtain
    \[
         \Pr{|\mysmalldot{\xi}{u_1}|,|\mysmalldot{\xi}{u_2}|\leq
         1/\sqrt{\delta d}}
         = O(1/(\rho^4\delta d))\mper
    \]
    The claim then follows by observing that by independence of the rows of $R$,
    \[\Pr{\norm{\infty}{Ru_1},\norm{\infty}{Ru_2} \leq 1/\sqrt{\delta d}} = \Big(\Pr{|\mysmalldot{\xi}{u_1}|,|\mysmalldot{\xi}{u_2}|\leq 1/\sqrt{\delta d}}\Big)^{h}\mper \qedhere \]
\end{proof}

\subsection{Summarizing Properties of Random Rademacher Kernel }
\label{summary}

\codinggadget*
\begin{proof}
    Fix any $\eps\in (0,1)$.
    Let $\delta :=  \eps/(3C_2),~h:= n^3,~d:= \ceil{\delta \cdot h},~k:=d(1+\eps)$, and let $N$ be the largest integer so that
    $h\geq d\log_{\sqrt{d}} (N/d)$. Let $R:=R_{h,N}$.
    By \cref{distance}, $d(\ker(R))\geq d(1-\eps)$ with probability $1-o(1)$.
    It is easily verified that for $n$ sufficiently large, $N\geq d^{C_2/\eps}$.
    Thus by \cref{overlap}, $d_2(\ker(R))\geq 2(1-\eps)d$, and so $\ker(R)$ is $2(1-\eps)$-non-overlapping with probability $1-o(1)$.

    We are left with verifying (Weak Local Density).
    First we note that $|\ker(R_{h,N})\cap H^N_{(1+\eps)d}|\geq (N/2)^{\eps d}$ by \cref{scnd:mmnt}.
    We then appeal to a probabilistic version of the Sauer-Shelah lemma, due to Micciancio~\cite{micciancio2001shortest}, which states that
    random projection of a sufficiently large subset of a hypercube slice, to a sufficiently low dimension, must cover the entire hypercube.
    \begin{theorem}[Theorem 5.9 of \cite{micciancio2001shortest}, Probabilistic Sauer-Shelah Lemma]
        For any $k,n,N\in \N$ and any $t > 0$, let $\calF \subseteq \{0,1\}^N$ be a set of at least $k! N^{4\sqrt{k}n/t}$ vectors, each with $k$ non-zero entries. If
        $T \in \{0,1\}^{n\times N}$ is chosen by setting each entry to $1$ independently at random with probability $p = 1/(4kn)$, then the probability that all of
        $\{0,1\}^n$ is contained in $T(\calF ) = \{Tx ~|~ x \in \calF\}$ is at least $1-6t$.
    \end{theorem}
    We apply the above theorem with the substitution $t\leftarrow 1/n^{0.1},~ n\leftarrow n,~k\leftarrow k,~N\leftarrow N$. It is easily checked that the assumptions of our claim imply that $k! N^{4\sqrt{k}n/t} < (N/2)^{\eps d}$, and so the application of the above theorem is valid.

    We conclude that with probability $1-o(1)$, $\ker(R),T$ as chosen above form a $(\frac{1+\eps}{1-\eps},2(1-\eps),n)$-coding gadget.
    Taking $\eps$ sufficiently small completes the proof of \Cref{coding:gad}.
\end{proof}

\bibliographystyle{alpha}
\bibliography{refs}

\appendix
\section{NP-Hardness of Exactly Solving  Quadratic Equations}
\label{sec:quadeq}

In this section, we prove a strengthened version of \Cref{prop:quadeq}, namely, \Cref{prop:quadeq} with an additional distinguished-coordinate property that is useful for hardness of NCP.
\begin{proposition}(NP-Hardness of Quadratic Equations)~\\
\label{prop:quadeq_dis}
Let $\F$ be any field. Given a system of quadratic equations over $\F^n$ of the form $\{Q_\ell(xx^T) = b_{\ell}\}_{\ell \in [m]}$
(resp. $\{Q_\ell(xx^T) = 0\}_{\ell \in [m]}$), it is NP-hard to distinguish between the following two cases:
\begin{itemize}
        \item (YES) There exists $x \in \{ 0, 1 \}^n$ with $x_n=1$, satisfying all $m$ equations.

        \item (NO) There does not exist $x \in \F^n$ (resp. $x \in \F^n\setminus \{0\}$)  satisfying all $m$ equations.
\end{itemize}
The distinguished-coordinate property refers to the fact that in the (YES) case, the solution $x$ satisfies $x_n=1$.
\end{proposition}

\begin{proof}
    We first show hardness of the homogeneous version $\{Q_\ell(xx^T) = 0\}_{\ell \in [m]}$.
    We reduce from Circuit-SAT problem, where the input is a Boolean circuit which consists of input gates as well as AND, OR, NOT gates with fan-in (at most) two and fan-out unbounded, and the goal is to find a Boolean assignment of its input gates that makes the output gate true. It is a prototypical NP-complete problem, since the Cook-Levin theorem is sometimes proved on Circuit-SAT instead of 3SAT (see e.g., Lemma 6.10 in \cite{AB09}).

    Given a Circuit-SAT instance $C$ with $n$ gates $y_1,\ldots,y_n$ (including input gates and logic gates), we build $n+1$ variables $\{x_1,\ldots,x_n,z\}$, and add the following equations:
    \begin{itemize}
        \item $x_i(x_i-z)=0, \forall i \in[n]$;
        \item for each AND gate $y_k = y_i \land y_j$ in $C$, an equation $x_k^2 = x_i x_j$;
        \item for each OR gate $y_k = y_i \lor y_j$ in $C$, an equation $z^2 - x_k^2 = (z-x_i) (z-x_j)$;
        \item for each NOT gate $y_k = \lnot y_i$ in $C$, an equation $z^2 - x_k^2 = x_i^2$;
        \item for $y_k$ being the output gate, an equation $z^2 = x_k^2$.
    \end{itemize}

    \paragraph{Completeness.}
    Let $y_1,\ldots,y_n \in \{0,1\}$ be an assignment to the gates of $C$ that makes the output true. It's easy to verify $$\left\{\begin{aligned}
        z & = 1, \\
        x_i & = y_i, &  \forall i \in [n]
    \end{aligned}\right.$$
    is a solution to the system of quadratic equations.

    \paragraph{Soundness.}
    If $z=0$, then by the first set of equations, each $x_i$ has to be 0 and this is an all-zero solution. Otherwise, every $x_i$ lies in $\{0,z\}$, and setting each $y_i=x_i z^{-1}$ is a satisfying assignment of $C$ since it satisfies every gate in $C$ and ensures that the output gate is true.

    The distinguished-coordinate property follows by the fact that $z=1$ in the (YES) case. We can set $z$ to be the last variable in the system.

    By simply replacing the variable $z$ in the above proof with a constant 1, we obtain hardness of the non-homogeneous version $\{Q_\ell(xx^T) = b_{\ell}\}_{\ell \in [m]}$, where in the soundness case there is no solution in all of $\F^n$ (as opposed to just $\F^n\setminus \{0\}$).
\end{proof}

\section{Tensoring}
\label{appendix:tensoring}
Here we prove the following fact about the distance of a tensor code.
\tensor*
\begin{proof}
    The LHS is at most the RHS since we may consider the element $uv^T\in U\otimes V$,
    where we choose $u\in U\setminus \{0\}$ that has sparsity $d(U)$ and $v\in V\setminus \{0\}$ that has sparsity $d(V)$.

    For the other direction, consider any nonzero $M\in U\otimes V$.
    There must be some non-zero entry in $M$, and so there is at least one nonzero column. Since this column lies in $U$, it must have at least $d(U)$ nonzero
    entries, and therefore at least $d(U)$ rows of $M$ are nonzero. Each such row lies in $V$ and hence has $d(V)$ nonzero entries. We conclude that $M$ has at least $d(U)\cdot d(V)$ nonzero entries.
\end{proof}

\section{Non-Overlap for Subspaces over $\F_q$}
\label{sec:overlap}

In this section, we prove \Cref{lem:overlap}.
\overlap*
\begin{proof}
    For $u,v$ being two linearly independent elements in $C$, let $m$ be the number of coordinates such that $u_i \neq 0$ or $v_i \neq 0$ but not both, and let $m'$ be the number of coordinates such that $u_i \neq 0$ and $v_i \neq 0$.
    Clearly,
    $$m+2m' \ge 2d(C).$$
    Since there are only $q-1$ choices of values for $u_i/v_i$ (for $u_i,v_i\neq 0$), there must exist $\lambda \neq 0$
    so that the vector $u -\lambda v$ has at most $m+m'-\frac{m'}{q-1}$ non-zero entries. This implies
    $$m+m'-\frac{m'}{q-1}\ge d(C).$$
    Multiplying the first inequality by $\frac{1}{q}$, the second by $\frac{q-1}{q}$, and adding, gives $m+m' \ge (1+\frac{1}{q})d(C)$ as desired.
\end{proof}

\section{Inapproximability of NCP}
\label{sec:ncp}

In this subsection, we prove the following inapproximability of the Nearest Codeword Problem.

\begin{theorem}
\label{thm:ncp}
Fix any finite field $\F_q$. No polynomial-time algorithm can given an affine subspace $V \subseteq \F_q^n$ and $s \in \N$, distinguishes between the following cases
\begin{compactenum}[]
    \item (YES) there exists nonzero $x \in V \cap \{ 0, 1 \}^n$ with $\norm{0}{x} \leq s$;
    \item (NO) every $x \in V \setminus \{ 0 \}$ satisfies $\norm{0}{x} \geq \gamma \cdot s$,
\end{compactenum}

\medskip

\begin{compactenum}[(a)]
    \item assuming $\NP\neq \classP$ when $\gamma>1$ is any constant;
    \item assuming $\NP\not\subseteq \classfont{DTIME}(2^{\log^{O(1)} n})$ when $\gamma =2^{\log^{1-\eps}n}$ for any fixed $\eps>0$;
    \item assuming $\NP\not\subseteq \bigcap_{\delta>0} \classfont{DTIME}(2^{n^{\delta}})$ when $\gamma=n^{c/\log \log n}$ for some fixed $c>0$.
\end{compactenum}
\end{theorem}

We prove \Cref{thm:ncp} by first slightly modifying the MDP($\F_q$) reduction in \Cref{subsec:mdpred} to have a distinguished coordinate, then giving a gap-preserving reduction from such MDP($\F_q$)  to NCP.

\begin{proposition}[MDP($\F_q$) Hardness with a Distinguished Coordinate]~\\
\label{prop:mdp_dist}
    The hardness of MDP($\F_q$) in \Cref{thm:mdp} holds even with the guarantee that the solution $x\in V \setminus \{0\}$ in the (YES) case satisfies $x_n=1$.
\end{proposition}

\begin{proof}
    Let $\eps = \frac{1}{9q}$. Note that the hardness of homogeneous quadratic equations holds even with a distinguished coordinate (\Cref{prop:quadeq_dis}). We thus append a distinguished coordinate to the reduction in \Cref{subsec:mdp_red}:
    \begin{equation}
        V:= \{(GXG^T,X_{n,n}) :~ Q_1(X)=0,\dots , Q_m(X)=0,~
        X^T=X,~X\in\mathbb{F}_q^{n\times n}\},
    \label{hred}
    \end{equation}
    where $G \in \F_q^{N \times n}$ is the generator matrix of an $\eps$-balanced code of minimum distance $d$.

    This reduces an instance of homogeneous quadratic equations with $n$ variables to an MDP$(\F_q)$ instance with $N^2+1=\text{poly}\left(n,\frac{1}{\varepsilon}\right)$ variables.

    For completeness, let $x \in \mathbb \{0,1\}^{n}$ be a non-zero solution to the system \Cref{hq} that satisfies $x_n=1$. Then $((Gx)(Gx)^T,1)\in V$ and has Hamming weight at most
    $\left(1+\frac{1}{3q}\right)d^2+1$.

    For soundness, note that if $(GXG^T,X_{n,n})$ is non-zero, then
    $X$ must be non-zero.
    The remainder of the analysis proceeds
    identically to the soundness analysis in \Cref{thm:mdp}, and we conclude any $Y \in V \setminus \{0\}$ has $\|Y\|_0 \ge \left(1+\frac{1}{q}\right) d^2$.

    Finally we amplify the gap using tensoring (\Cref{tensoring}),
    and we note that the distinguished coordinate property of the YES case is preserved under tensoring.
\end{proof}

Given \Cref{prop:mdp_dist}, \Cref{thm:ncp} is proved as follows.

\begin{proof}[Proof of \Cref{thm:ncp}]
    Consider the reduction from MDP($\F_q$) with a distinguished coordinate to NCP, given by mapping a subspace $V \subseteq \mathbb F_q^{n}$ to the affine subspace $V':=\{x\in V:x_n=1\}$.

    Completeness follows from the distinguished-coordinate property. Soundness follows by noting that any $x\in V'$ is a non-zero vector in $V$, with the same sparsity.
\end{proof}

We finish this section with two remarks on NCP.

\begin{remark}
A gap-preserving Cook-reduction from MDP to NCP was known in the literature (see e.g., \cite{GMSS99} and Exercise 23.13 in \cite{GRS18}). However, we have a Karp-reduction here thanks to the distinguished-coordinate property.

We also remark that there are non-trivial ways to amplify gap for NCP (see e.g., Theorem 22 of \cite{LLL24} and Section 4.2 of \cite{BGKM18}). Thus one can start with the NP-hardness of non-homogeneous quadratic equations, perform the reduction in this section to get a mild constant gap for NCP, and then amplify the gap. This is another way to prove \Cref{thm:ncp}.
\end{remark}

\begin{remark}
    The relatively near codeword problem parameterized by $\rho\in (0,\infty)$ (denoted by $\text{RNC}^{\,(\rho)}$) is a promise problem
    defined as follows. Given a subspace $V\subseteq \mathbb{F}_q^n$ of an (unknown) minimum distance $d$, a vector $b\in \mathbb{F}_q^n$ and an integer $t$ with the promise that $t<\rho\cdot d$, the task is to find a codeword in $V$ of Hamming distance at most $t$ from $b$.

    Our reduction for NCP yields NP-Hardness of $\text{RNC}^{\,(1+\varepsilon)}$
    for any fixed $\varepsilon>0$ (it also yields hardness for the gap version of RNC for an appropriately small constant gap). \cite{dumer2003hardness} show
    hardness of $\text{RNC}^{\,(\rho)}$ for any $\rho>1/2$, albeit under randomized reductions.
\end{remark}

\end{document}